\input{epsf}

\documentclass[conference]{IEEEtran}

\usepackage{epsf}
\usepackage{graphicx}
\usepackage[cmex10]{amsmath}
\usepackage{amsthm}
\usepackage{amssymb}
\usepackage{epsfig,latexsym,amsmath,epsf,amssymb,amsfonts}
\usepackage{algorithm, algorithmic}
\usepackage{placeins}
\usepackage{cite}
\usepackage{comment} 
\usepackage{subfigure}
\usepackage{lipsum}

\usepackage[top=0.75in, bottom=1in, left=0.625in, right=0.625in]{geometry}
\usepackage{bbm}

\renewcommand{\theequation}{\arabic{equation}}
\newtheorem{Definition}{\hskip 0pt Definition}
\newtheorem{theorem}{\hskip 0pt Theorem}
\newtheorem{Lemma}{\hskip 0pt Lemma}
\newtheorem{Proposition}{\hskip 0pt Proposition}

\theoremstyle{remark}
\newtheorem{Remark}{\hskip 0pt Remark}
\begin{document}
%
\title{Power Allocation with Stackelberg Game in Femtocell Networks: A Self-Learning Approach\vspace*{-3mm}}



\author{
\IEEEauthorblockN{Wenbo Wang\IEEEauthorrefmark{1},
Andres Kwasinski\IEEEauthorrefmark{1} and
Zhu Han\IEEEauthorrefmark{2}}
\IEEEauthorblockA{\IEEEauthorrefmark{1}Department of Computer Engineering, Rochester Institute of Technology, NY, USA}
\IEEEauthorblockA{\IEEEauthorrefmark{2}Department of Electrical and Computer Engineering, University of Houston, TX, USA}\vspace*{-7mm}}


\maketitle \vspace{-3in}

\begin{abstract}
This paper investigates the energy-efficient power allocation for a two-tier, underlaid femtocell network. The behaviors of the Macrocell 
Base Station (MBS) and the Femtocell Users (FUs) are modeled hierarchically as a Stackelberg game. The MBS guarantees its own QoS requirement
by charging the FUs individually according to the cross-tier interference, and the FUs responds by controlling the local transmit power 
non-cooperatively. Due to the limit of information exchange in intra- and inter-tiers, a self-learning based strategy-updating mechanism is 
proposed for each user to learn the equilibrium strategies. In the same Stackelberg-game framework, two different scenarios based on the 
continuous and discrete power profiles for the FUs are studied, respectively. The self-learning schemes in the two scenarios are designed 
based on the local best response. By studying the properties of the proposed game in the two situations, the convergence property of the 
learning schemes is provided. The simulation results are provided to support the theoretical finding in different situations of the proposed 
game, and the efficiency of the learning schemes is validated.
\end{abstract}

\IEEEpeerreviewmaketitle
\section{Introduction}
Femtocells are considered to be an efficient and economic solution to enhance the indoor experience of the cellular mobile users 
\cite{6171992}. A femtocell is a low-power, short-range access point, which can be quickly deployed by the end-users. It provides better 
spatial reuse of the spectrum by serving the nearby users who have poorer connections with the Macrocell Base Station (MBS) due to 
penetration loss. In practice, the femtocell network usually operates underlaying the macrocell network. This is mainly due to the ad-hoc 
topology of the femtocells and thus the lack of coordination between the MBS and Femto Access Points (FAPs). Consequently, 
inter-cell/cross-tier interference arises, and interference mitigation becomes necessary for preventing performance deterioration. 

Due to the ad-hoc topology of the femtocells, the FAP deployment faces the limited information exchange both across tiers and among the 
femtocells. Therefore, it is desirable that the interference management of the femtocells is fully distributed, and each Femtocell User (FU) 
is capable of adapting to the surrounding environment with minimum information. With this in mind, we study the power control 
schemes for a shared-spectrum, two-tier femtocell network. We note that the Macrocell User (MU) prefers that the cross-tier 
interference is minimized, while the FUs prefers to transmit with the best Signal-to-Interference-plus-Noise-Ratio (SINR). Considering that 
private objectives contradict with each other, it is natural to introduce the tools of game theory and model the cross-tier, self-centric 
interactions in the framework of non-cooperative games. 

\subsection{Related Work}
Under the framework of non-cooperative games, the early study \cite{983324} has discovered that power control purely based on the 
non-cooperative games will lead to inefficient equilibria. In order to obtain the Pareto-preferred equilibria, a number of approaches 
including the introduction of repetition \cite{965533} or externalities (e.g., pricing) \cite{1337220, 5382529} are adopted in the research.
As shown by studies in non-hierarchical networks \cite{1337220, 5382529, 1523423}, choosing a proper pricing mechanisms with respect to 
different utility functions can be an efficient way of determining the desired properties of the equilibria.

When it comes to the resource allocation in hierarchical networks, such as the femtocell networks and cognitive radio networks, the Stackelberg 
game \cite{han2012game} based modeling is widely preferred since it is able to reflect the features of hierarchy and ad-hoc topology in the 
network \cite{Zhang:2009:SGU:1530748.1530753, 6171995}. the Stackelberg game is characterized by the sequential decision making manner 
(namely, the follower-leader strategy updating), and hence suitable for modeling the heterogeneous user behaviors in the network. Due to the 
computational intensity for obtaining the Stackelberg Equilibrium (SE), most of the existing studies \cite{6171995,5659869,5935038} adopt a
utility model that favors the derivation of a closed-form SE, and solve for the SEs through transforming the games as hierarchical 
optimization problems. 

Although the optimization-technique-based methods are able to precisely analyze the properties of the SEs, their scope is limited to 
the games with a certain categories of utility functions. Beyond those games, a natural idea is to resort to tools of 
iterative learning in repeated games for searching the SEs. A body of literature on non-hierarchical networks can be found applying 
iterative strategy-learning methods \cite{6117762, 5683552}, usually based on the assumption of the discrete strategy space \cite{4278411}.
However, since these learning methods assume homogeneous behaviors among the players, most of their application to hierarchical
networks are also limited within an uniform learning model \cite{6542770, 6654861}.

\subsection{Self-Learning under Pricing}
In this paper, we model the power allocation problem in the two-tier femtocell network from the perspective of the Stackelberg game. In the 
game, the MBS behaves as the leader and controls the total cross-tier interference by setting prices to each FU-FAP link. The FU-FAP links 
behave as the followers to optimize their energy efficiency through interactive power allocation. In designing the distributed, hierarchical 
power control scheme with pricing, the power efficiency is adopted as the utility of the FU-FAP link. we investigate the scenarios of the 
continuous and discrete power profile of the femtocells, respectively. Due to the computational complexity of obtaining the 
closed-form solution of the equilibrium, we investigate the property of the game in the two situations, and propose two self-centric 
strategy-learning algorithms for the follower game based on the local myopic best response, respectively. We provide the theoretical 
proof of the convergence of the learning schemes. Also, we provide two heuristic algorithms for obtaining the optimal MBS 
prices in the corresponding situations. In the simulation, we provide the experimental comparison on the network performance between the two 
scenarios, and demonstrate the efficiency of the proposed learning algorithms.

\section{Problem Formulation}
\subsection{Network Model}
We consider the uplink transmission of a two-tier femtocell network with a single MBS and $K$ FAPs. The MBS and the FAPs share the same 
bandwidth $W$ and each of them is scheduled to serve one single user at each time instance. The MBS is required to keep the 
femtocell-to-macrocell interference to an acceptable level. Due to the ad-hoc topology of the femtocells, we assume that the information 
exchange only happens between the MBS and FAPs. For analytical tractability, we suppose that all the channels involved are 
block-fading and remains constant during each transmission block. 

In what follows, we let $0$ denote the index of the MU-MBS pair and $k\in\mathcal{K}=\{1, \ldots, K\}$ denote the index of a FU-FAP pair. 
The channel power gain between the transmitter of pair $i$ and the receiver of pair $j$ is denoted by $h_{i,j}$, 
where $i,j\in\mathcal{K}\bigcup\{0\}$. The power of transmitter $k$ is denoted by $p_k$ and the power vector of all the FUs is denoted by
$\mathbf{p}=[{p}_1,\ldots, {p}_k]$. The noise variance for the transmitter-receiver pair $k$ is denoted by $N_k$. Then the SINR level at the
MBS can be expressed as:
\begin{equation}
 \label{eq1}
 \gamma_0(p_0,\mathbf{p})=\frac{h_{0,0}p_0}{N_0+\sum\limits_{k\in\mathcal{K}}h_{k,0}p_k},
\end{equation}
and the SINR level at the $k$-th FAP can be expressed as
\begin{equation}
 \label{eq2}
 \gamma_k(p_0,\mathbf{p})=\frac{h_{k,k}p_k}{N_k+h_{0,k}p_0+\sum\limits_{j\in\mathcal{K}\backslash\{k\}}h_{j,k}p_j}.
\end{equation}

During the operation, it is usually beneficial to shift some calls served by the MBS to the FAP. Therefore, we 
suppose that for the MBS, the requirement on the femtocell-to-macrocell interference is not rigid. Instead, the MU transmit with a fixed 
power and the MBS charges each FU-FAP link for causing interference with a certain price to control the interference level. We denote the 
vector of interference prices at a time interval by $\pmb\lambda\!=\![\lambda_1,\ldots, \lambda_K]$, in which $\lambda_k$ is the price for 
unit interference caused by FU $k$. The goal of the MBS is to maximize the total revenue of collecting payments from the FU-FAP links:
\begin{eqnarray}
 \label{eq3}
 \begin{array}{ll}
  \max\limits_{\pmb\lambda}\left(u_0\!=\!\sum\limits_{k\in\mathcal{K}}\lambda_kh_{k,0}p_k\right).
 \end{array}
\end{eqnarray}

For simplicity, in what follows we use the terms FU and the FU-FAP link interchangeably. For the FUs, we assume that each local transmit 
power $p_k$ is limited by the physical constraint $p_k\in[0, {p}^{\max}_{k}]$. The goal of FU $k$ is to maximize its local net payoff by 
adapting ${p}_k$:
\begin{equation}
 \label{eq4}
  \max\limits_{0\le{p}_k\le{p}^{\max}_{k}}\bigg(u_k=\psi(\gamma_k, {p}_k)-\lambda_sh_{k,0}p_k\bigg),
\end{equation}
in which $\psi(\gamma_k, {p}_k)$ is the utility function of FU $k$. Considering the practical scenarios, we adopt 
the local utility $\psi(\gamma_k, p_k)$ as the energy efficiency (namely, the data received per unit energy consumption) of each FU:
\begin{equation}
 \label{eq9} 
 \psi(\gamma_k, \mathbf{p}_k)=\frac{W\log(1+\gamma_k)}{p_k+ p_a},
\end{equation}
where $p_a$ denotes the additional circuit power consumption for FU $k$. We assume that the local SINR $\gamma_k$ can be perfectly measured 
at the FAP. For the proposed femtocell network, we suppose that the following assumptions hold:
\begin{itemize}
 \item[i)] $p_k$ is significantly greater than $p_a$;
 \item[ii)] The femtocell-to-femtocell interference is sufficiently small.
\end{itemize}
Assumption i) is based on the fact that most power is consumed during operation by the amplifier and the radio transceiver 
\cite{5783984}. Assumption ii) is based on the practical concerns that (a) the FAPs are of low power, so the peak transmit 
power is limited; and (b) the intercell gains between femtocells are usually weak due to the path loss (since the indoor penetration loss 
is usually significant). Based on the assumptions, we assume that the SINR constraint for a FU-FAP link is negligible.

\section{Stackelberg Game Analysis}
We model the user interactions in the proposed network as a hierarchical game with the MBS and the FU-FAP links choosing their actions in a 
sequential manner. When the power allocation of the FUs are given as $\mathbf{p}$, we can define the leader game from the 
perspective of the MBS as:
\begin{equation}
 \label{eq5}
  \mathcal{G}_l={\langle}\pmb{\Lambda}, \{u_0(\pmb\lambda, \mathbf{p})\}{\rangle},
\end{equation}
in which the MBS is the only player of the game and the action of the player is the price vector $\pmb\lambda\in\pmb\Lambda$. 

When the leader action is given by $\pmb{\lambda}$, we can define the non-cooperative follower game from the perspective of the FUs as:
\begin{equation}
 \label{eq6}
  \mathcal{G}_f={\langle}\mathcal{K},\mathcal{P}, \{u_k(\pmb\lambda, \mathbf{p}_k, \mathbf{p}_{-k})\}_{k\in\mathcal{K}}{\rangle},
\end{equation}
in which FU $k$ is one player with the local action $\mathbf{p}_k\in\mathcal{P}_k$.

With each player behaving rationally, the goal of the game (\ref{eq5}) and (\ref{eq6}) is to finally reach the Stackelberg Equilibrium (SE) 
in which both the leader (MBS) and the followers (FUs) have no incentive to deviate. We first assume that the strategies of the FUs and MBS 
are continuous. Then the SE of the game can be mathematically defined as follows:
\begin{Definition}[Stackelberg Equilibrium]
\label{def_se}
 The strategy $(\pmb\lambda^*, \mathbf{p}^*)$ is a SE for the proposed Stackelberg game (\ref{eq5}) and (\ref{eq6}) if
 \begin{eqnarray}
 \label{eq7}
  & u_0(\pmb\lambda^*, \mathbf{p}^*){\ge}u_0(\pmb\lambda, \mathbf{p}^*), \forall \pmb\lambda{\in}\pmb\Lambda,\\
  \label{eq8}
  & u_k(\pmb\lambda^*, \mathbf{p}^*_k, \mathbf{p}^*_{-k}){\ge}u_k(\pmb\lambda^*, \mathbf{p}_k, \mathbf{p}^*_{-k}), \forall k\in\mathcal{K}, 
  \forall \mathbf{p}_k\in{\mathcal{P}_k}.  
\end{eqnarray}
\end{Definition}

\subsection{Femtocell Power Allocation with Continuous Strategies}
We start the analysis of the Stackelberg game in (\ref{eq5}) and (\ref{eq6}) by back induction. Suppose that the MBS first sets its strategy 
as $\pmb\lambda$, then we obtain the non-cooperative follower subgame described by (\ref{eq6}). In order to show the existence of a 
pure-strategy Nash Equilibrium (NE) in $\mathcal{G}_f$, we introduce the concept of supermodular game as follows.
\begin{Definition}[Supermodularity\cite{han2012game}]
 \label{de_id}
 A function $f:\mathcal{X}\times\mathcal{T}\rightarrow \mathbb{R}$ is said to have increasing differences (supermodularity) in $(x,t)$ if for 
 all $x'\ge x$ and $t'\ge t$,
 \begin{equation}
 \label{eq_sm}
  f(x',t')-f(x,t')\ge f(x',t)-f(x,t).\nonumber 
 \end{equation}
\end{Definition}

\begin{Definition}[Supermodular game \cite{han2012game, 1198610}]
 \label{def_sg}
 A general normal-form game ${\langle}\mathcal{N}, \{\mathcal{S}_i\}_{i\in\mathcal{N}}, \{u_i\}_{i\in\mathcal{N}}{\rangle}$ is 
 a supermodular game if for any player $i\in\mathcal{N}$,
 \begin{itemize}
  \item[i)] the strategy space $\mathcal{S}_i$ is a compact subset of $\mathbb{R}^K$.
  \item[ii)] the payoff function $u_i$ is upper semi-continuous in $\mathbf{s}=(\mathbf{s}_i,\mathbf{s}_{-i})$.
  \item[iii)] $u_i$ is supermodular in $\mathbf{s}_i$ and has increasing difference between any component of $\mathbf{s}_i$ and any component of 
  $\mathbf{s}_{-i}$.
 \end{itemize}
\end{Definition}  

The supermodular property of the proposed follower subgame (\ref{eq6}) is given by Theorem \ref{Thm_SG}.
\begin{theorem}
 \label{Thm_SG}
 Given the strategy $\pmb\lambda$ of the MBS, the follower subgame (\ref{eq6}) is a supermodular game if $\gamma_k\ge{p_a}/{p_k}$.
\end{theorem}
\begin{proof}
 See Appendix A.
\end{proof}

Given the opponent strategies $p_{-k}$, we define the local Best Response (BR) of FU $k$ by
\begin{equation}
 \label{eq10}
  \hat{p}_k({p}_{-k})=\mathop{\arg\max}\limits_{0\le{p}_k\le{{p}^{\max}_{k}}}\bigg(u_k=\psi(\gamma_k, p_k)-\lambda_sh_{k,0}p_k\bigg).
\end{equation}
Then Theorem \ref{Thm_SG} immediately yields Proposition \ref{Prop1} \cite{1198610}:
\begin{Proposition}
 \label{Prop1}
 At least one pure-strategy NE exists in the follower subgame (\ref{eq6}) and the following points hold:
 \begin{itemize}
  \item [i)] The set of NEs (\ref{eq8}) has the component-wise greatest element $\overline{\mathbf{p}^*}$ and least element 
  $\underline{\mathbf{p}^*}$.
  \item [ii)] If the BRs are single-valued, and each FU uses the BR starting from the smallest (largest) elements of the strategy space to 
  update their strategies, then the strategies monotonically converge to the smallest (largest) NE.
  \item[iii)] If the game has a unique NE, then with any arbitrary initial strategy, the local myopic BRs converges to the NE.
 \end{itemize}
\end{Proposition}

The properties of $\mathcal{G}_f$ in Proposition \ref{Prop1} sheds light on the solution to the local strategy-learning scheme for the FUs. 
To take advantages of these properties, we first examine the BRs in $\mathcal{G}_f$ and obtain Lemma \ref{le_single_value}:
\begin{Lemma}
 \label{le_single_value}
 Given the MBS strategy $\pmb\lambda$, $u_k$ is strictly quasiconcave and the best-response $\hat{p}_k$ is single-valued for each FU.
\end{Lemma}
\begin{proof}
 See Appendix B.
\end{proof}
Based on Lemma \ref{le_single_value}, we can further prove that the BR of FU $k$ has the following properties, and therefore is a standard 
function \cite{han2012game}:
\begin{itemize}
 \item [i)] positivity: $\hat{p}_k>0$,
 \item [ii)] monotonicity: for any $p'_{-k},p_{-k}\in\mathcal{P}_{-k}$, if $p'_{-k}>p_{-k}$, then $\hat{p}_k(p'_{-k})\ge\hat{p}_k(p_{-k})$,
 \item [iii)] scalability: $\forall \alpha>1$, $\alpha\hat{p}_k(p_{-k})>\hat{p}_k(\alpha p_{-k})$.
\end{itemize}
Based on the aforementioned properties, we can directly apply the results in \cite{han2012game} and deduce the uniqueness of the pure-strategy NE 
in the follower subgame $\mathcal{G}_f$.
\begin{theorem}
 \label{Thm_SNE}
 Given any MBS strategy $\pmb\lambda$, the follower subgame (\ref{eq6}) has a unique NE if the following condition is satisfied
 \begin{equation}
 \label{eq11}
 \frac{h_{k,k}p_k}{N_k+h_{0,k}p_0+\sum_{j\in\mathcal{K}}h_{j,k}p_k}\ge\frac{p_a}{p_k}.
\end{equation}
\end{theorem}
\begin{proof}
 See Appendix C.
\end{proof}

\begin{Remark}
 If in subgame $\mathcal{G}_f$ the condition of Theorem \ref{Thm_SNE} is satisfied, the condition of Theorem \ref{Thm_SG} will also be 
 satisfied. (\ref{eq11}) indicates that to ensure the uniqueness of the NE, the SINR of a FU should be significantly larger than the ratio 
 between its transmit power and circuit power. Such a condition is guaranteed by our assumption of the network.
\end{Remark}

We assume that the channel power gain $h_{k,0}$ is known and the SINR $\gamma_k$ can be perfectly measured by each FAP. Then, based on Lemma 
\ref{le_single_value}, $\hat{p}_k$ can be solved locally with the bisection method \cite{boyd2004convex}. Based on Proposition \ref{Prop1} and 
Theorem \ref{Thm_SNE}, the asynchronous strategy-updating mechanism defined in \cite{1198610} can be directly applied to $\mathcal{G}_f$. By 
Proposition \ref{Prop1}, the  convergence to the NE is guaranteed from any arbitrary initial power vector. The strategy-learning 
algorithm is summarized as Algorithm 1:
\begin{algorithm}[htb]
 \label{alg1}
 \caption{Asynchronous strategy updating}
 \begin{algorithmic}[1]
 \REQUIRE 
 each FU sets up an infinite increasing time sequence $\{T^i_k\}_{k\in\mathcal{K}}$ for scheduling strategy update.
 \FORALL {$t\in\{T^i_k\}_{k\in\mathcal{K}}$}
  \FORALL {$k$ s.t. $t=\{T^i_k\}$}
  \STATE given ${p}^{t-1}_{-k}$, obtain $p_k^t=\hat{p}_k(p^{t-1}_{-k})$ as in (\ref{eq10}) with bisection.
  \ENDFOR
 \ENDFOR
 \end{algorithmic}
\end{algorithm}

\subsection{Approximate Solution to the Price of MBS}
When considering the leader subgame $\mathcal{G}_l$, we assume that the strategies of the FUs are given as $\mathbf{p}$ from (\ref{eq10}). For 
the subgame (\ref{eq5}), the local BR is given by
\begin{equation}
 \label{eq13}
 \hat{\pmb\lambda}=\arg\max\limits_{\pmb\lambda{\succeq}0}\left(\sum_{k\in\mathcal{K}}\lambda_kh_{k,0}p_k \right).
\end{equation}

To investigate the solution of (\ref{eq13}), we first consider the feasible region of $\lambda_k$. We note from (\ref{eq10}) that the maximum 
value of $u_k$ is lower-bounded by $0$ (when $p_k\!=\!0$) and upper-bounded by $\psi(\tilde{\gamma}_k, \tilde{p}_k)$, in which $(\tilde{\gamma}_
k, \tilde{p}_k)=\arg\max\psi({\gamma}_k,{p}_k)$. Thereby, the price $\lambda_k$ charged by the MBS is also upper-bounded. Otherwise, if 
$\lambda_k$ is too high (i.e., making $\psi(\tilde{\gamma}_k, \tilde{p}_k)\!<\!\lambda_sh_k{p}_k$), FU $k$ will stop transmitting 
and be forced out of the game. With the aforementioned bound on $u_k$, we look for the constraint on  $\lambda_k$ in (\ref{eq10}). 
However, since $u_k$ in (\ref{eq10}) is a transcendental function, it is difficult to derive a closed-form 
expression of the constraint on $\lambda_k$. Then, the challenge of analyzing $\mathcal{G}_l$ is to find an efficient way for 
obtaining the optimal price $\hat{\pmb\lambda}$.

By jointly investigating the leader and the follower subgames, we can show that a finite, optimal price $\hat{\lambda}_k$ for each FU coexists 
with the NE of the follower subgame.
\begin{theorem}
 \label{Them_SE}
 In the Stackelberg game defined by (\ref{eq5}) and (\ref{eq6}), at least one pure-strategy SE with finite price vector $\hat{\pmb
 \lambda}$ from the leader game exists.
\end{theorem}
\begin{proof}
The proof is based on Theorem 3.2 of \cite{han2012game}. Given the condition that each local strategy in the game is compact and convex and 
the corresponding payoff function is quasiconcave, the existence of a pure-strategy SE is guaranteed. For the FUs, quasiconcavity of 
$u_k(\mathbf{p},\pmb\lambda)$ in $p_i$ is given by Lemma \ref{le_single_value}. For the MBS, $u_0$ is an affine function of $\pmb\lambda$, 
hence being quasiconcave in $\pmb\lambda$. It is trivial that $\mathbf{p}$ and $\pmb\lambda$ are convex and compact, then based on Theorem 
3.2 of \cite{han2012game}, there exists a pure-strategy NE in the game. Beyond the discussion after (\ref{eq13}) on the fact that 
$0\le\lambda_k<\infty$, we can derive the relationship of the BRs between the FUs and the MBS from $\frac{\partial u_k}{\partial p_k}=0$ as:
\begin{equation}
\label{eq_thm_se_1} 
\tilde{\lambda}_k=\frac{W\tilde{G}_k}{h_k (1+\tilde{\gamma}_k)(\tilde{p}_k+p_a)}-\frac{W\log(1+\tilde{\gamma}_k)}{h_k(\tilde{p}_k+p_a)^2} ,
\end{equation}
in which $\tilde{p}_k$ and $\tilde{\gamma}_k$ are the local BR and the corresponding SINR. $\tilde{G}_k$ is given by (\ref{eq_thm2_2}) in 
Appendix A. For (\ref{eq_thm_se_1}), the value of the right-hand side expression is upper-bounded since $0\le\tilde{p}_k\le{p}^{\max}_k$. 
Then $\tilde{\lambda}_k$ is finite. \end{proof}

Our approximate solution to the leader subgame is inspired by the pioneering work of \cite{1337220}, which models the asymptotic behaviors of 
the equilibriu bhum power vector and the corresponding payments. We assume that each FU's behavior can be 
asymptotically modeled by two regions, the price-insensitive region and the price-sensitive region. In the price-insensitive region, the FU's 
behaviors are hardly influenced by the price. In the price-sensitive region, the local power allocation $p_k$ is driven toward 0. 
Mathematically, the two-region model can be expressed by the following asymptotes:
\begin{itemize}
 \item Low-price asymptote as $\lambda_k\rightarrow0$:
 \begin{equation}
  \label{eq12}
  \left\{
  \begin{array}{ll}
    \gamma^*_k\approx\hat{\gamma}_k, \\
    r_k(\lambda_k, p_k)=h_{k,0}\hat{p}_k\lambda_k\propto\lambda_k.
   \end{array}\right. 
 \end{equation}
 \item High-price asymptote as $\lambda_k\rightarrow\infty$:
  \begin{equation}
  \label{eq14}
    p_k\!\approx\!\frac{W}{\lambda_k({N_k+h_{0,k}p_0})}-p_a. 
 \end{equation}
\end{itemize}
In (\ref{eq12}), $\tilde{\gamma}_k$ is the equilibrium SINR when $\lambda_k\!=\!0$. The details for deriving (\ref{eq12}) and (\ref{eq14}) is 
presented in Appendix D. Since in the low-price asymptote, $r_k$ increases with $\lambda_k$ and in the high-price asymptote it decreases with
$\lambda_k$, the maximum payment must happen between the two regions. Then, we can extend the two FU payment asymptotes toward each other 
until they meet at the intersection price $\lambda^a_k$. With such an approximation, $r_k(\lambda^a_k,p_k)$ will be the maximum payment 
received from FU $k$. Combining (\ref{eq12}) and (\ref{eq14}), we can obtain the intersection point for the two asymptotes as:
  \begin{equation}
  \label{eq15}
   \lambda^a_k\approx\frac{W}{({p}^*_k+p_a)(N+h_{0,k}p_0)},
 \end{equation}
in which ${p}^*_k$ is the power allocation corresponding to the SINR $\gamma^*$ in (\ref{eq12}). It can be obtained by setting $\pmb\lambda=0$
and solving (\ref{eq_prop_abm_1}) with the BR-based asynchronous strategy-updating mechanism.

\subsection{Femtocell Power Allocation in Discrete Strategies}
We continue to extend the analysis of the game to the scenario in which the FUs choose their strategies from a finite, discrete set of powers. 
In this case, the conditions for NEs (i.e., Theorems \ref{Thm_SNE} and \ref{Them_SE}) are not satisfied anymore. Therefore, the properties 
of the NE need to be re-evaluated. Within the same game structure of (\ref{eq5}) and (\ref{eq6}), we denote the action set of the FUs 
by $\mathcal{P}_k\!=\!\{p_k^1\!=\!0,\ldots,p_k^{|\mathcal{P}_k|}\}$. It is well known that every finite non-cooperative game has a 
mixed-strategy NE \cite{han2012game}. For the follower subgame, we define the mixed-strategies of FU $k$ as $\pmb\pi_k\!=\![\pi^1_k, \ldots, 
\pi^{|\mathcal{P}|}_k]$, in which $\pi_k^j(p_k^j)\!=\!\Pr(p_k\!=\!p_k^j)$ is the probability for FU $k$ to choose the $j$-th action 
$p^j_k\in{\mathcal{P}_k}$. Then, given any MBS price $\pmb\lambda$, there will be at least one mixed-strategy NE for the FUs. Different from 
the continuous game, the expected net payoff of FU $k$ becomes
\begin{equation}
  \label{eq16}
   u_k(\pmb\pi_k,\pmb\pi_{-k},\pmb\lambda)\!=\!\sum_{\mathbf{p}\in\mathcal{P}}(\psi_k(p_k,p_{-k})\!-\!\lambda_kh_{k,0}p_k)
   \prod_{i\in\mathcal{K}}\prod_{1\le{j}\le|\mathcal{P}_k|}
   \pi^j_i,
\end{equation}
in which $\sum_{j}\pi_k^j\!=\!1, 0\!\le\!\pi_k^j\!\le\!1$. Similarly, the expected revenue of the MBS becomes
\begin{equation}
  \label{eq17}
   u_0(\pmb\lambda,\pmb\pi)=\sum_{k\in{\mathcal{K}}}\lambda_k\sum_{1\le{j}\le|\mathcal{P}_k|}\pi^j_kh_{k,0}p^j_k.
\end{equation}

Due to the limit of information exchange, FU $k$ can only attain its local payoff $u_k(p^j_k,\pmb\pi_{-k})$ each time when it chooses $p_k^j$ and 
the other FUs adopt the mixed strategies $\pmb\pi_{-k}$. To ensure that each FU is able to learn its Nash distribution $\pmb\pi_k$, we adopt the 
Logit best response function (namely, the smooth BR based on entropy perturbation) \cite{leslie2003convergent}:
\begin{equation}
 \label{eq18}
 \beta^t_k(p^j_k|\pmb\pi_{-i})\!=\!\frac{\exp(U^{t\!-\!1}_k(p^j_k,\pmb\pi_{-k})/\tau)}{\sum_{1\le{i}\le|\mathcal{P}_k|}\exp(U^{t\!-\!1}_k(p_k^i,
 \pmb\pi_{-k})/
 \tau)},
\end{equation}
in which $U^{t}_k$ is the estimated expected payoff at time $t$. $\tau$ is a positive scalar (also known as Boltzmann temperature) that controls 
the sensitivity of the BR to perturbation. 

Based on the two-timescale strategy learning scheme \cite{leslie2003convergent}, we introduce two coupled stochastic learning processes to 
approximate $U_k(p^j_k)$ and $\pi_k(p^j_k)$ in (\ref{eq18}) as follows:
\begin{align}
 \label{eq19}
 &{U}^t_k(p^j_k)\!=\!{U}^{t-1}_k(p^j_k)\!+\!\alpha_1^t\mathbbm{1}_{(\pi^j_k(p^j_k))}\left({u^t_k(p^j_k)\!-\!{U}^{t-1}_k(p^j_k)}\right),\\ 
 \label{eq20}
 &\pi^t_k(p^j_k)\!=\!\pi^{t-1}_k(p^j_k)\!+\!\alpha_2^t\left(\beta^t_k(p^j_k|\pmb\pi_{-i})-\pi^{t-1}_k(p^j_k)\right).
\end{align}
In (\ref{eq19}) and (\ref{eq20}), $u^t_k(p^j_k)$ is the instant payoff observation at the FAP in (\ref{eq4}) and $\beta^t_k(p^j_k|
\pmb\pi_{-i})$ is the smooth BR (\ref{eq18}). $\mathbbm{1}_{(\pi^j_k(p^j_k))}$ is the indicator function. $\mathbbm{1}_{(\pi^j_k(p^j_k))}=1$ 
if $\pi^j_k(p^j_k)=1$ and otherwise $\mathbbm{1}_{(\pi^j_k(p^j_k))}=0$. The parameter sequence $\alpha_1^t$ and $\alpha_2^t$ satisfy the 
following conditions:
\begin{eqnarray}
 \label{eq21}
 \left\{
 \begin{array}{ll}
  \lim\limits_{T\rightarrow0}\sum\limits^T_{t=1}\alpha_1^t=+\infty, \qquad \lim\limits_{T\rightarrow0}\sum\limits^T_{t=1}(\alpha_1^t)^2<+\infty,\\
  \lim\limits_{T\rightarrow0}\sum\limits^T_{t=1}\alpha_2^t=+\infty, \qquad \lim\limits_{T\rightarrow0}\sum\limits^T_{t=1}(\alpha_2^t)^2<+\infty,\\
  \lim\limits_{t\rightarrow0}\frac{\alpha_2^t}{\alpha_1^t}=0.
 \end{array}\right.
\end{eqnarray}
The conditions in (\ref{eq21}) ensures that the learning of strategies changes on a slower timescale than that of the action values. Based on 
the discussion in \cite{ECTA:ECTA376}, we can show that the learning processes converges by Theorem \ref{Thm_Converge}:
\begin{theorem}
 \label{Thm_Converge}
 With any arbitrary $\pmb\pi^0$ and $U_k^0$, the strategy-learning mechanism (\ref{eq18})-(\ref{eq20}) almost surely converges to some fixed 
 point. The probability of converging to a NE is non-zero.
\end{theorem}
\begin{proof}
 See Appendix E.
\end{proof}

In the scenario of finite strategies, it is even more difficult to obtain a SE point for the MBS prices. However, by investigating the property 
of concavity in the payoff function for any element $\pi^j_k$ or $\lambda_k$ of the joint strategy vector $(\pmb\pi,\pmb\lambda)$, we can show 
that there exists at least one SE with MBS price in pure-strategy: 
\begin{theorem}
 \label{Thm_FSE}
 With the discrete set of $\mathcal{P}_k$, the Stackelberg game with the payoff functions (\ref{eq16}) and (\ref{eq17}) has a SE composed of 
 the mixed-strategy power allocation $\hat{\pmb\pi}$ and the pure-strategy price $\hat{\pmb\lambda}$, which is finite in each 
 $\hat{\lambda}_k$. 
\end{theorem}

\begin{proof}
The proof for the existence of the SE follows directly from Theorem 1 of \cite{1965Rosen}. It is easy to see that the payoff function 
$U_k(\pmb\pi,\pmb\lambda)$ is linear (hence concave) in $\pi^i_k$ if $\pmb\lambda$ and the rest of the elements of $\pmb\pi$ are fixed. 
Similarly, $U_0(\pmb\pi,\pmb\lambda)$ is linear in $\lambda_k$ with $\lambda_{-k}$ and $\pmb\pi$ fixed. Therefore, following the discussion 
of Theorem 1 in \cite{1965Rosen} and the Kakutani fixed point theorem, there exists a fixed point $(\pmb\pi^*,\pmb\lambda^*)$ satisfying 
Definition \ref{def_se}.
 
The proof of a finite $\lambda_k$ in the SE is similar to that in Theorem \ref{Them_SE}. If we assume that $(\pmb\lambda, \pmb\pi)$ is a SE, 
then from (\ref{eq16}) we obtain $\forall k\in\mathcal{K}, 1\!\le\!j\!\le\!|\mathcal{P}_k|$,$\frac{\partial u_k}{\partial \pi^j_k}\!=\!0$, which
is equivalent to
\begin{equation}
 \label{eq_pr_thm_fse1}
 \sum_{p_{-k}}\psi(p^j_k,p_{-k})\prod_{m\ne k,i}\pi_m^i=\lambda_kh_{k,0}p_k^j.
\end{equation}
Similar to the continuous-game scenario in Theorem \ref{Them_SE}, it is easy to verify that as $\lambda_k\!\rightarrow\!\infty$, 
$\pi^1_k\!\rightarrow\!1$ and $u_0\!\rightarrow\!0$. We note that $\forall \pmb\lambda$, the equation array (\ref{eq_pr_thm_fse1}) always has a solution 
since the mixed-strategy NE exists. With (\ref{eq17}) and (\ref{eq_pr_thm_fse1}), we can obtain
\begin{equation}
 \label{eq_pr_thm_fse2}
 u_0=\sum_k\sum_{p_k^j\in{\mathcal{P}_k}}\frac{\psi(p_k^j,p_{-k})\pi_k^j\prod_{m\ne k,i}\pi_m^i}{h_{k,0}p_k^j},
\end{equation}
which must have a non-zero maximum value. Therefore, we can always find a finite $\lambda_k$ that maximize (\ref{eq17}) with the NE of the 
follower game, Otherwise, it will contradict with the fact that $u_0\rightarrow0$ as $\lambda_k\rightarrow\infty$.
\end{proof}

Following our discussion, we propose a pricing mechanism based on the myopic best response as Algorithm 2:
\begin{algorithm}[htb]
 \label{alg2}
 \caption{Heuristic price updating}
 \begin{algorithmic}[1]
 \REQUIRE  
 The MBS sets $\lambda_k\!=\!0$ and the FUs arbitrarily initialize $\pmb\pi_k^0$.
 \WHILE{the cross-tier SINR requirement (\ref{eq1}) is not met}
  \WHILE{not converged}
    \STATE $\forall k\in\mathcal{K}$, FU $k$ updates $\pmb\pi_k$ with (\ref{eq18})-(\ref{eq20}).
  \ENDWHILE
  \STATE $\forall k\in\mathcal{K}$, FU $k$ report $\pmb\pi_k$ to the MBS.
  \STATE The MBS announces the prices $\lambda_k$:
  \begin{equation}
  \lambda_k=\frac{\sum_{\mathbf{p}\in\mathcal{P}}\psi_k(p_k,p_{-k})\prod_{i\in\mathcal{K},j}\pi^j_i}{\sum_jh_{k,0}p^j_k\pi^j_k},
  \end{equation}
 \ENDWHILE
 \end{algorithmic}
\end{algorithm}

\section{Simulation Results}
The objective of this section is to provide insight into the impact of pricing on the network performance at the equilibrium, and the 
influence of strategy discretization on the learning process. In the simulation, we assume that the FAPs are randomly located indoors within 
a circle centered at the MBS with a radius of $300$m. Each FU is placed within a circle centered at the 
corresponding FAP with a radius of $15$m. The channel gains of the transmitter-receiver pairs are generated by a lognormal shadowing pathloss
model with $h_{i,j}=d^{-k}_{i,j}$, in which $k$ is the pathloss factor, $k=4$ for the FUs and $k=2.5$ for the MU. The parameters used
in the simulation are summarized in Table \ref{t1}.\vspace{-2mm}
\begin{table} [!ht]
  \caption{Main Parameters Used in the Femtocell Network Simulation}
  \label{t1}
  \centering
 \begin{tabular}{ l|r}
  \hline
  Parameter &  Value \\
  \hline
  Shared Bandwidth $W$ & 1MHz\\
  \hline
  Maximum MU transmit power ${p}^{\max}_0$ & $27$dBm\\
  \hline
  Feasible region for FU transmit power $[{p}^{\min}_k, {p}^{\max}_k]$ & $[0, 20]$dBm\\
  \hline
  Additional FU circuit power $p_a$ & $3$dBm\\
  \hline
  AWGN power $N_k, k=0,\ldots,K$ & $-40$dBm\\
  \hline
  SINR threshold of the MU & 3dB\\
  \hline
\end{tabular}
\end{table}\vspace{-2mm}

\subsection{Analysis of the Equilibrium in the Continuous Game}
In the first simulation, we study the influence of the MBS price $\pmb\lambda$ on the equilibrium of the follower subgame with 6 FAPs. 
For the convenience of demonstration, we suppose that the MBS charges an identical price to each FU-FAP link. Figure 
\ref{fig1} provides the payoffs evolution of both the MBS and the FU-FAP links at the follower-game NE as the uniform price increases. 
We note in Figure 
\ref{fig1}.b that there exists an optimal value of $\pmb\lambda$ to maximize the MBS revenue, which provides an experimental evidence of 
Theorem \ref{Them_SE}. We also note from Figure \ref{fig1}.a and \ref{fig1}.b that there exists a plateau region in which the average power 
efficiency remains almost the same while the MBS revenue keeps increasing. It means that without undermining the social welfare of the 
femtocells, the MBS is able to control the cross-tier interference by choosing the price within the plateau region.
\begin{figure}[!t]
\centering
\includegraphics[width=0.44\textwidth]{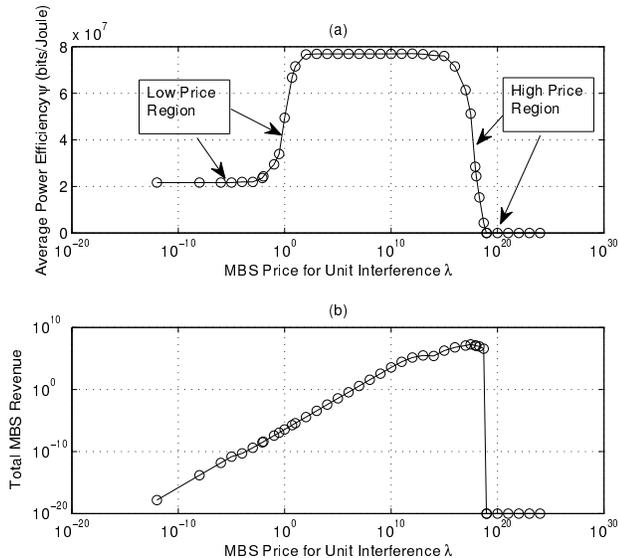}
\caption{Influence of the unit interference price on the NE of the continuous follower-game. (a) Average power efficiency at the NE vs. the 
unit price $\lambda$; (b) total revenue collected by the MBS at the NE vs. the unit price $\lambda$.}
\label{fig1}
\end{figure}

We note from Figure \ref{fig1} that at the optimal MBS price (i.e., the SE point), the FU performance dramatically 
degrades from the optimal condition. This leaves the room in practice for the MBS to trade a portion of the revenue for a socially better
performance. Following the first simulation, we investigate the network performance at the proposed price (\ref{eq15}) as the number 
of FAPs varies. In Figures \ref{fig2} and \ref{fig3}, the user performance at the proposed price is compared to that at the accurate SE price, 
and that with no price ($\lambda\!=\!0$). The accurate SE price is obtained using a semi-exhaustive searching method with bisection, and the 
utilities are obtained from Monte Carlo simulation. Figure \ref{fig2} shows that at the proposed price a better FU performance can be achieved 
(Figure \ref{fig2}.b) at the cost of losing a significant portion of the MBS revenue (Figure \ref{fig2}.a). However, by measuring the 
expected SINR of the MU in Figure \ref{fig3}, we note that such trade-off is worthwhile since the performance deterioration of the MU is small 
when compared to the gain of the FU performance. Again, Figure \ref{fig2}.b and Figure \ref{fig3} shows the fact that with no externality, the
network performance can be heavily impaired (see curve ``No price, $\lambda\!=\!0$'').

\begin{figure}[!t]
\centering
\includegraphics[width=0.44\textwidth]{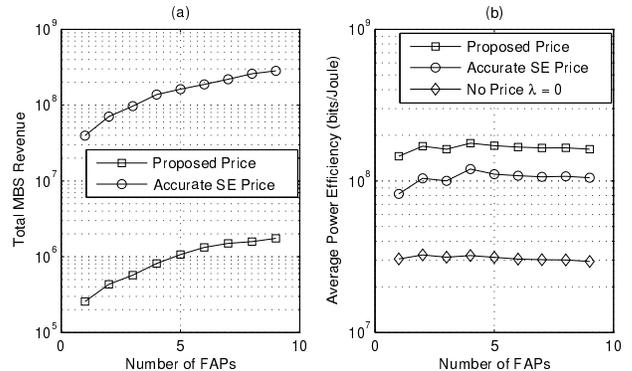}
\caption{MBS revenue and FU power efficiency at the SEs. (a) MBS revenue vs. the number of FAPs; (b) FU power efficiency vs. the
number of FAPs.}
\label{fig2}
\end{figure}

\begin{figure}[!t]
\centering
\includegraphics[width=0.43\textwidth]{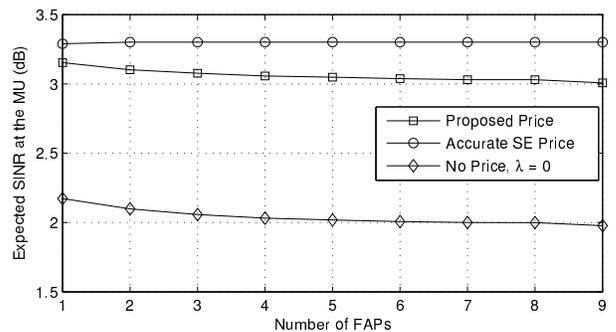}
\caption{Expected equilibrium SINR at the MU vs. the number of FAPs.}
\label{fig3}
\end{figure}

\subsection{Analysis of the Equilibrium in Discrete Game}
Since we are interested in the impact of strategy discretization on the network performance, we adopt the same network setup in the first 
simulation for the discrete game. The parameter for the self-learning algorithm is given in Table \ref{t2}. To compensate for the divergence 
caused by the self-learning algorithm (Theorem \ref{Thm_Converge}), we examine the expected user utilities with Monte Carlo simulation. 
For the convenience of demonstration, we also suppose that the MBS places an uniform price. The simulation results are shown in Figure 
\ref{fig4}.
\begin{table} [!ht]
  \caption{Main Parameters Used in the Self-Learning Algorithm}
  \label{t2}
  \centering
 \begin{tabular}{ l|r}
  \hline
  Parameter &  Value \\
  \hline
  Boltzmann temperature $\tau$ & 1\\
  \hline
  Learning rate for $U^t_k$ and $\pmb\pi^t_k$ & ${1}/{t}$ and ${1}/{t^2}$\\
  \hline
  Number of candidate power $M$ & $M=6$\\
  \hline
  Power sampling equation & $p^j_k=(1-\frac{j}{M}){p}^{\min}_k+\frac{j}{M}{p}^{\max}_k$\\
  \hline
\end{tabular}
\end{table}

Figure \ref{fig4}.b shows the existence of an optimal $\pmb\lambda$ that maximizes the MBS revenue. It can be interpreted as an experimental 
evidence of Theorem \ref{Thm_FSE}. Comparing Figure \ref{fig4}.a and Figure \ref{fig1}.a, we note that
in the discrete follower game (Figure \ref{fig4}.a), the ``plateau'' region extends to $\lambda\!\rightarrow\!0$. It means that different from 
the case of continuous game (Figure \ref{fig1}.a), lacking an external price does not severely undermine the social performance of the FUs. However, the femtocell network may suffer 
from discretization of the power space and only achieve approximately $1/2$ of the performance in the continuous game at most of the NEs.
\begin{figure}[!t]
\centering
\includegraphics[width=0.43\textwidth]{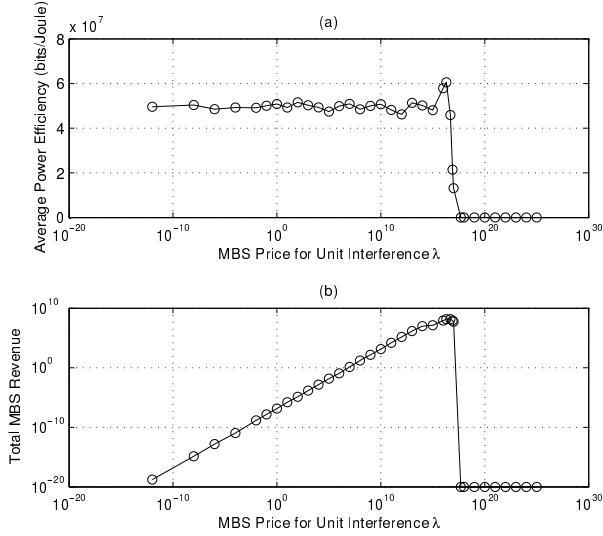}
\caption{Influence of the unit interference price on the NE of the discrete follower-game. (a) (Expected) average power efficiency at the NE 
vs. the unit price $\lambda$; (b) total revenue collected by the MBS at the NE vs. the unit price $\lambda$.}
\label{fig4}
\end{figure}

Figure \ref{fig5} shows the user performance at the proposed price with Algorithm 2, the accurate SE price and zero price, respectively. 
The comparison in Figure \ref{fig5} shows that the FU performance at the proposed price is the best of the three. However, from Figure 
\ref{fig5}.b we note that without a pricing scheme, the FUs can still achieve a good performance level. As the number of FAPs increases, we 
can observe a deterioration in the FU performance. It means that mixed-strategies NE in the discrete game is not as good as the 
continuous-game NE in maintaining the performance when the size of network increases (see Figure \ref{fig2}.b).

\begin{figure}[!t]
\centering
\includegraphics[width=0.44\textwidth]{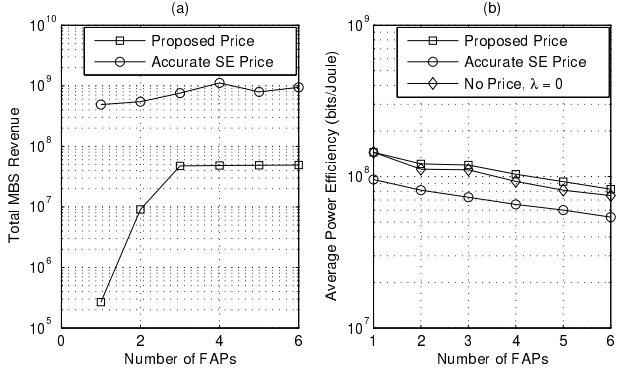}
\caption{Expected MBS revenue and FU power efficiency at the SEs in the discrete strategy space. (a) MBS revenue vs. the number of FAPs; (b) 
FU power efficiency vs. the number of FAPs.}
\label{fig5}
\end{figure}

Finally, we demonstrate in Figures \ref{fig6} and \ref{fig7} the convergence of the self-learning algorithm in the discrete game of 
$6$ FUs. Figure \ref{fig6} shows a snapshot of FUs' transmit-power evolution during the learning process of Algorithm 2. Figure \ref{fig7} 
shows the corresponding strategy evolution of FU 1. By running the simulation for multiple times, we observe that most of the learning 
processes converge within $600$ iterations.

\begin{figure}[!t]
\centering
\includegraphics[width=0.45\textwidth]{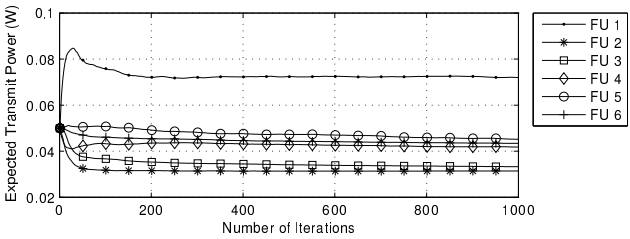}
\caption{Expected transmit power vs. the number of iterations.}
\label{fig6}
\end{figure}

\begin{figure}[!t]
\centering
\includegraphics[width=0.45\textwidth]{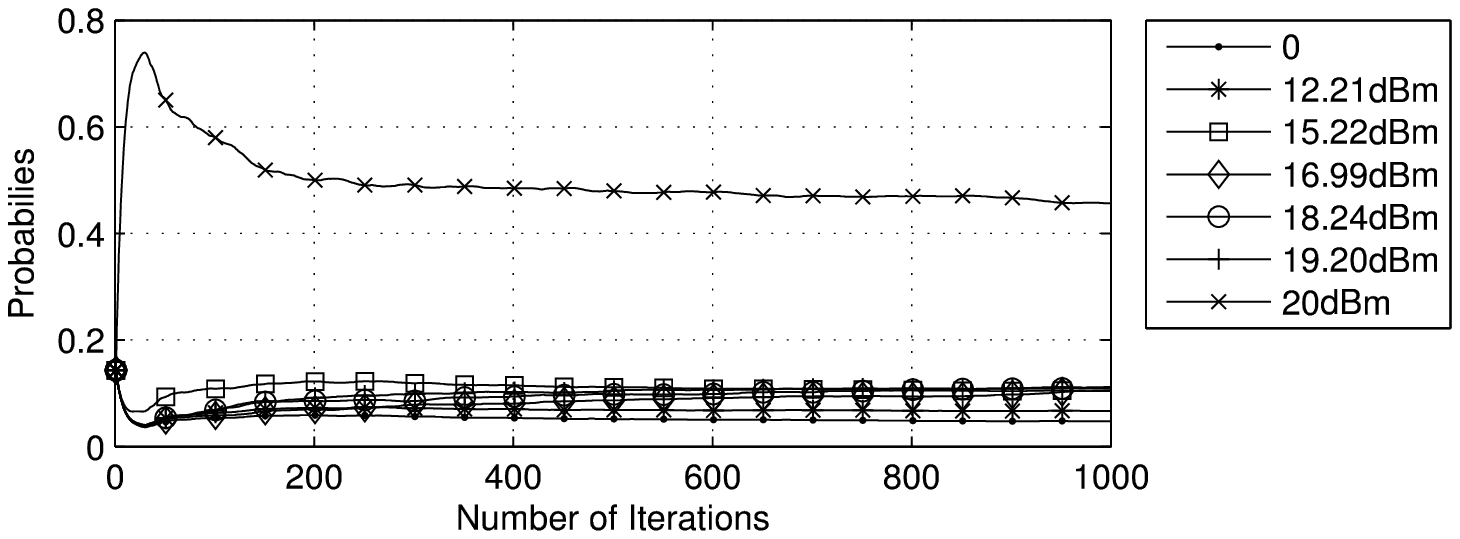}
\caption{Probability of power selection vs. the number of iterations by FU 1.}
\label{fig7}
\end{figure}

\section{Conclusion}
In this paper, we have formulated the price-based power allocation problem in the two-tier femtocell network under the framework 
of the Stackelberg game. We have provided the theoretical analysis of the properties of the equilibria in the scenarios of continuous FU 
power space and discrete power space, respectively. We have proposed two self-sufficient learning algorithms, one for each situation, for 
learning the NE of the follower game with limited information exchange. We have also provided the theoretical proof for the convergence of the 
learning algorithms. Our simulation results provides important insight in the different impact of the pricing mechanism on the network 
performance in the scenarios of the continuous game and the discrete game. In the simulation, we also show the efficiency of the proposed 
heuristic price-searching mechanism in the game. Our study provides an alternative way of designing the resource allocation protocols since
no intercell information exchange is required in the femtocells.

\appendices
\section{Proof of Theorem 1}
\begin{Lemma}[\!\cite{1198610}]
 \label{le_td}
 If a function $f(\mathbf{s})$ is twice differentiable, then supermodularity is equivalent to $\frac{\partial^2f(\mathbf{s})}{\partial{\mathbf{s}_i}
 \partial{\mathbf{s}_j}}\ge0$  $\forall \mathbf{s}_i, \mathbf{s}_j$, $j\ne i$.
\end{Lemma}

The first and the second conditions of a supermodular game in Definition \ref{def_sg} (trivially) holds for the proposed follower subgame 
(\ref{eq6}). Then, Theorem \ref{Thm_SG} can be derived based on Lemma \ref{le_td}. By taking the component-wise derivative of $u_k({p}_k, 
\mathbf{p}_{-k})$ with respect to $p_k$, we obtain:
\begin{equation}
 \label{eq_thm2_1}
 \frac{\partial u_k}{\partial p_k}\!=\!-\frac{W\log(1\!+\!\gamma_k)}{(p_k\!+\!p_a)^2}\!+\!\frac{WG_k}{(1\!+\!\gamma_k)(p_k\!+\!p_a)}\!
 -\!\lambda_kh_{k,0},
\end{equation}
in which
\begin{equation}
 \label{eq_thm2_2}
 G_k=\frac{h_{k,k}}{N_0+h_{0,k}p_0+\sum_{i\in{\mathcal{K}}\backslash\{k\}}h_{i,k}p_i}.
\end{equation}
Then $\forall k, j\in\mathcal{K}$, $k\ne j$, the value of $\frac{\partial^2 u_k}{\partial p_k\partial p_j}$ is given by:
\begin{align}
 \label{eq_thm2_3}
 \frac{{\partial}^2 u_k}{{\partial} p_k{\partial} p_j}=&\frac{WH_kp_k}{(1\!+\!\gamma_k)(p_k+p_a)^2}\!+
 \!\frac{WH_k\gamma_k}{(1\!+\!\gamma_k)^2(p_k+p_a)} \nonumber \\
 &-\!\frac{WH_k}{(1\!+\!\gamma_k)(p_k+p_a)}, 
\end{align}
in which 
\begin{equation}
 \label{eq_thm2_4}
 H_k=\frac{h_{j,k}h_{k,k}}{\left(N_0+h_{0,k}p_0+\sum_{i\in{\mathcal{K}}\backslash\{k\}}h_{i,k}p_i\right)^2}. 
\end{equation}
It is easy to verify from (\ref{eq_thm2_3}) that $\frac{{\partial}^2 u_k}{{\partial} p_k{\partial} p_j}\ge 0$ if $\gamma_k\ge{p_a}/{p_k}$. By
Lemma \ref{le_td}, $u_k$ has increasing difference between ${p}_k$ and any component of ${p}_{-k}$ if $\gamma_k\ge{p_a}/{p_k}$. By Definition
\ref{def_sg}, the proof of Theorem \ref{Thm_SG} is completed. 

\section{Proof of Lemma 1} 
The proof of Lemma \ref{le_single_value} is derived from investigating the quasiconcavity of the FU payoff functions. For conciseness, the 
readers are referred to Sections 3.1 and 3.4 of \cite{boyd2004convex} for the details of the definition in superlevel set and quasiconcavity.

\begin{Lemma}[\!\cite{boyd2004convex}]
  \label{le_global_op}
 Suppose $f: \mathcal{D}\rightarrow\mathbb{R}$ be strictly quasiconcave where $\mathcal{D}\subset\mathbb{R}^N$ is convex. Then any local 
 maximum of $f$ on $\mathcal{D}$ is also a global maximum of $f$ on $\mathcal{D}$. Moreover, the set $\arg\max\{f(x)|x\in\mathcal{D}\}$ is 
 either empty or a singleton.
\end{Lemma}

We first show that the utility function $u_k$ in $\mathcal{G}_f$ is quasiconcave. We examine the $\alpha$-superlevel set of 
$u_k(p_k, p_{-k})$ in ${p}_k$, which is equivalent to the 0-superlevel set of $f_{\alpha}(p_k,p_{-k})$:
\begin{equation}
  \label{eq_def_prov2}
  \begin{array}{ll}
  \mathcal{P}_{k, 0}\!=\!\big\{{p}_k\big\vert f_{\alpha}(p_k, p_{-k})\ge 0, 0\le{p}_k\le{{p}^{\max}_{k}},\\
  f_{\alpha}(p_k, p_{-k})\!=\!W\log(1\!+\!\gamma_k)\!-\!(\lambda_kh_{k,0}p_k+\alpha)(p_k+p_a) \big\}.
  \end{array}
\end{equation}
We note that $f_{\alpha}(p_k, p_{-k})$ is a concave function, so $\mathcal{P}_{k,\alpha}$ is convex by the definition of convexity. 
By the definition of quasiconcavity, $u_i(\pmb\lambda^*,{p}_k,{p}_{-k})$ is quasiconcave in ${p}_k$. 

Then we show that $u_i({p}_k,{p}_{-k})$ is strictly quasiconcave in ${p}_k$ so the BR is a global maximum and thus single-valued. Without 
loss of generality, we consider the power allocation $\hat{p}_k\!\in\![0, {p}^{\max}_k]$ with $u_k(\hat{p}_k, p_{-k})\!=\!\alpha$. We assume 
that a different power allocation $\tilde{p}_k$ satisfies $u_k(\tilde{p}_k, p_{-k})\!\ge\!u_k(\hat{p}_k, p_{-k})$. Correspondingly, 
$f_{\alpha}(\hat{p}_k, p_{-k})\!=\!0$ and $f_{\alpha}(\tilde{p}_k, p_{-k})\!\ge\!0$. Observing the following condition for $f_{\alpha}(\hat{p}_k,
p_{-k})=0$:
\begin{equation}
 \label{eq_def_prov3}
 W\log(1+G_kp_k)=(\lambda_sh_{k,0}p_k+\alpha)(p_k+p_a),
\end{equation}
in which $G_k$ is given in (\ref{eq_thm2_2}). We note that the right-hand side of (\ref{eq_def_prov3}) is a strictly increasing concave function 
and the left-hand side is a strictly increasing convex function in $[0, {p}^{\max}_k]$. Then the solution to $f_{\alpha}(\hat{p}_k, p_{-k})\!=
\!0$ is unique in $[0, {p}^{\max}_k]$, so $f_{\alpha}(\tilde{p}_k, p_{-k})\!>\!f_{\alpha}(\hat{p}_k, p_{-k})$. Based on the definition of 
concave function, the following inequality also holds for 
$0<\delta<1$: 
\begin{equation}
 \label{eq_def_prov4}
 \begin{array}{ll}
 f_{\alpha}(\delta\hat{p}_k+(1-\delta)\tilde{p}_k)&\ge{\delta}f_{\alpha}(\hat{p}_k)+(1-\delta)f_{\alpha}(\tilde{p}_k)\\
 &>f_{\alpha}(\hat{p}_k)=0. 
 \end{array}
\end{equation}
Therefore, the condition for strict quasiconcavity holds as $u_k(\delta\hat{p}_k+(1-\delta)\tilde{p}_k)>\min\left(u_k(\tilde{p}_k, 
p_{-k}),u_k(\hat{p}_k, p_{-k})\right)=\alpha$. Then Lemma \ref{le_single_value} is a direct conclusion based on Lemma \ref{le_global_op}.

\section{Proof of Theorem 2}
\begin{Lemma}[\!\!\!\cite{han2012game}]
 \label{le_signle_ne}
 If the best-response functions of a non-cooperative game $\mathcal{G}$ are standard functions for all the players, then the game has a 
 unique NE in pure strategies.
\end{Lemma}
Observing (\ref{eq4}), we note that the maximum net-payoff function is lower-bounded by 0 with $p_k\!=\!0$. Since the power vector is always 
nonnegative, the property of positivity in the BR for each FU $k$ immediately follows Lemma \ref{le_single_value}.

We denote $I_k(p_{-k})\!=\!(N_k\!+\!h_{0,k}p_0\!+\!\sum_{j\in{\mathcal{K}}\backslash\{k\}}h_{j,k}p_j)/h_{k,k}$. Noting that $I_k(p_{-k})$ is 
a strictly increasing function of $p_{-k}$, monotonicity of $\hat{p}_k(p_{-k})$ can be illustrated by proving that function $p_k(I_k)$ is 
monotonically increasing in $I_k$. From (\ref{eq_thm2_1}) we obtain the necessary condition for $p_k$ to be the BR as $\frac{\partial u_k}
{\partial p_k}\!=0\!$, which is equivalent to
\begin{align}
 \label{eq_th_3_2}
 \omega(p_k, I_k)=&\frac{W(p_k+p_a)}{I_k}-W(1+\frac{p_k}{I_k})\log(1+\frac{p_k}{I_k}) \nonumber\\
 &-\lambda_kh_{k,k}(p_k+p_a)^2(1+\frac{p_k}{I_k})=0. 
\end{align}
Since $\frac{\partial{p_k}}{\partial{I_k}}=-\frac{\partial{\omega}}{\partial{I_k}}/\frac{\partial{\omega}}{\partial{p_k}}$, we have
\begin{equation}
 \label{eq_th_3_3}
 \frac{\partial{\omega}}{\partial{I_k}}\!=\!\frac{1}{I^2_k}\left(\xi(p_k)\!+\!Wp_k\log(1\!+\!\frac{p_k}{I_k})\!-\!Wp_a\right), 
\end{equation}
in which $\xi(p_k)=\lambda_kh_{k,k}(p^2_ap_k+2p_ap_k^2+p_k^3)$, and
\begin{equation}
 \label{eq_th_3_4}
 \frac{\partial{\omega}}{\partial{p_k}}\!=-\frac{1}{I_k}\left(\zeta(p_k)+{W\log(1\!+\!\frac{p}{I_k})}\right)，
\end{equation}
in which $\zeta(p_k)=\lambda_kh_{k,k}(p_a\!+\!p_k)(p_a\!+\!2I_k\!+\!3p_k)$. We note that $\frac{\partial{\omega}}{\partial{p_k}}<0$, then
the property of monotonicity holds iff $\frac{\partial{\omega}}{\partial{I_k}}\ge0$. With the inequality of logarithmic function 
\cite{topsok2006some}, $\log(1+x)\ge x/(1+x)$ for $x\ge-1$, we obtain
\begin{equation}
 \label{eq_th_3_5}
 \frac{\partial{\omega}}{\partial{I_k}}\ge\frac{1}{I^2_k}\left(\xi(p_k)\!+\frac{W}{I_k+p_k}\left({p_k^2-p_a(I_k+p_k)}\right)\right).
\end{equation}
Therefore, $\frac{\partial{p_k}}{\partial{I_k}}\ge0$ if $p_k^2-p_a(I_k+p_k)\ge0$. Then we obtain the condition for $\hat{p}_k(p_{-k})$ to be 
monotonic as:
\begin{equation}
 \label{eq_th_3_6}
 \frac{h_{k,k}p_k}{N_k+h_{0,k}p_0+\sum_{j\in\mathcal{K}}h_{j,k}p_k}\ge\frac{p_a}{p_k}.
\end{equation}

The proof of scalability is based on Lemma \ref{le_single_value}. According to Lemma \ref{le_single_value}, there is a one-to-one 
correspondence between $\hat{p}_k$ and $\hat{\gamma}_k$. We define $J_k(p_{-k})=\sum_{j\in{\mathcal{K}}\backslash\{k\}}h_{j,k}p_j$,
then from (\ref{eq2}) the BR can be written as
\begin{equation}
 \label{eq_th_3_7}
 \hat{p}_k(p_{-k})=\frac{\hat{\gamma}_k\left(N_0+h_{0,k}p_0+J_k(p_{-k})\right)}{h_{k,k}}.
\end{equation}
From $\frac{\partial u_k}{\partial p_k}\!=\!0$, we can prove that $\frac{\partial \gamma_k}{\partial J_k}\!\le\!0$ with the same technique as
proving monotonicity (which is omitted for conciseness). Therefore, $\gamma_k$ is a decreasing function of $J_k$. Since $J_k(p_{-k})$ is a 
standard function \cite{1198610}, we realize that if $\alpha>1$, $\hat{\gamma}_k(\alpha p_{-k})\le\hat{\gamma}_k(p_{-k})$ and $J_k
(\alpha p_{-k})\le\alpha J_k(p_{-k})$. Then, monotonicity holds for $\hat{p}_k(p_{-k})$ since
\begin{equation}
 \label{eq_th_3_9}
 \hat{p}_k(\alpha p_{-k})\le\frac{\hat{\gamma}_k(p_{-k})\left(N_0\!+\!h_{0,k}p_0\!+\!\alpha J_k(p_{-k})\right)}{h_{k,k}}\le\alpha 
 \hat{p}_k(p_{-k}).
\end{equation}
Therefore, $\hat{p}_k(p_{-k})$ is a standard function. Based on Lemma \ref{le_signle_ne}, the NE of the follower subgame (\ref{eq6}) is unique.

\section{The Derivation of Asymptotic Behavior Models}
The necessary condition for the NE of the FUs is given by (\ref{eq_thm_se_1}). As $\lambda_k\rightarrow0$, the solution of the BRs in the follower 
game will be independent of $\lambda_k$ and can be approximated by
\begin{equation}
\label{eq_prop_abm_1}
(1+\gamma_k)\log(1+\gamma_k)-W\gamma_k-WG_kp_a=0, k\in\mathcal{K},
\end{equation}
in which $G_k$ is given by (\ref{eq_thm2_2}). From Lemma \ref{le_single_value}, the solution to (\ref{eq_prop_abm_1}) is unique. Then the payment 
by FU $k$ will be a linear function of $\lambda_k$, $r_k=h_{k,0}\hat{p}_k\lambda_k$.

As $\lambda_k\rightarrow\infty$, $\forall k\in\mathcal{K},p_k\rightarrow0$. From (\ref{eq_thm_se_1}) we obtain:
\begin{equation}
\label{eq_prop_abm_2}
h_{k,k}\lambda_k(p_k+p_a)\!=\!\frac{Wh_{k,k}}{I_k\!+h_{k,k}\!p_k}\!-\frac{\!W\log(1+\frac{h_{k,k}p_k}{I_k})}{(p_k+p_a)},
\end{equation}
in which $I_k$ is the sum of interference plus noise defined in Appendix C. With $p_k\rightarrow0,\forall k$, (\ref{eq_prop_abm_2}) can be approximated by:
\begin{equation}
\label{eq_prop_abm_3}
h_{k,k}\lambda_k(p_k+p_a)\!\approx\!\frac{Wh_{k,k}}{N_k+h_{0,k}p_0}.
\end{equation}
From (\ref{eq_prop_abm_3}) we obtain 
\begin{equation}
\label{eq_prop_abm_4}
p_k\!\approx\!\frac{W}{\lambda_k({N_k+h_{0,k}p_0})}-p_a.
\end{equation}
Based on (\ref{eq_prop_abm_1}) and (\ref{eq_prop_abm_4}) we obtain the asymptotic models (\ref{eq12}) and (\ref{eq14}) for the FU behaviors.

\section{Proof of Theorem 4}

\begin{Lemma}[\!\!\!\cite{ECTA:ECTA376}]
 \label{le_transform}
 Consider game $\mathcal{G}$ with payoff function $u_k(\mathbf{s})$ for player $k$. If the sequence of stochastic fictitious play converges,
 it also holds for game $\tilde{\mathcal{G}}$ with payoff function $\tilde{u}_k(\mathbf{s})=\kappa_ku_k(\mathbf{s})+
 \vartheta(s_k)$, in which $\kappa_k$ is a positive constant and $\vartheta(s_k)$ only depends on player $k$'s own behavior.
\end{Lemma}

\begin{Lemma}[\!\!\cite{ECTA:ECTA376}]
 \label{le_fp_supermodular}
 Consider stochastic fictitious play starting from any arbitrary $\pmb\pi^0$. If $\mathcal{G}$ is a supermodular game, then
 \begin{equation}
  \label{eq_le_fp_sup}
  \Pr(\omega\{\pmb\pi^0_k\}\subset{RP} \textrm{ or } \omega\{\pmb\pi^0_k\}\subset M_i\cap[\underline{\pmb\pi}_k,\overline{\pmb\pi}_k] \textrm{ for } k)=1
  \nonumber
 \end{equation}
 where $\omega\{\pmb\pi^0_k\}$ is an invariant set of the solution trajectory starting from $\pmb\pi^0_k$. $RP$ is the set of rest points (fixed 
 points) and $\underline{\pmb\pi}_k, \overline{\pmb\pi}_k\in RP$ such that $RP\subset[\underline{\pmb\pi}_k,\overline{\pmb\pi}_k]$. $M_i$ is a finite 
 Lipschitz submanifold and every persistent non-convergence trajectory is asymptotic to one in $M_i$.
\end{Lemma}

The proof of Theorem \ref{Thm_Converge} starts by investigating the property of supermodularity in the FU subgame. With discrete power set, the 
expected payoff function (\ref{eq16}) can be rewritten as:
\begin{equation}
  \label{eq_pro_5_1}
   u_k(\pmb\pi_k,\pmb\pi_{-k},\pmb\lambda)\!=\!\sum_{\mathbf{p}\in\mathcal{P}}\psi_k(p_k,p_{-k})\prod_{i\in\mathcal{K},j}\pi^j_i-\sum_j
   \lambda_kh_{k,0}p^j_k\pi^j_k,
\end{equation}
which is in the form ${u}_k(\pmb\pi)\!=\!\kappa_k\tilde{u}_k(\mathbf{\pmb\pi})\!+\!\vartheta(\pmb\pi_k)$ with $\kappa_k\!=\!1$. By Lemma 
\ref{le_transform}, we only need to check the game with payoff
\begin{equation}
  \label{eq_pro_5_2}
   \tilde{u}_k(\pmb\pi_k,\pmb\pi_{-k},\pmb\lambda)\!=\!\sum_{\mathbf{p}\in\mathcal{P}}\psi_k(p_k,p_{-k})\prod_{i\in\mathcal{K},j}\pi^j_i.
\end{equation}
With Definition \ref{def_sg} and Lemma \ref{le_td}, it is easy to check that the game with payoff function (\ref{eq_pro_5_2}) and mixed strategies
is a supermodular game. Based on Theorem 6 of \cite{leslie2003convergent}, the process of $\pi_k^t$ produced by (\ref{eq18})-(\ref{eq20}) will 
almost surely be an asymptotic pseudotrajectory of the smooth BR dynamics:
\begin{equation}
  \label{eq_pro_5_3}
   \dot{\pi}_k=\beta_k(\pmb\pi_{-k})-\pi_k.\nonumber
\end{equation}
Then, based on Lemma \ref{le_fp_supermodular}, for the game with payoff (\ref{eq_pro_5_2}) stochastic fictitious play almost surely converges 
with any arbitrary initial $\pmb\pi^0$. With Lemma \ref{le_transform}, the convergence holds for the original subgame with payoff $u_k$, so 
Theorem \ref{Thm_Converge} is proved.


%
\bibliographystyle{IEEEtran}
\bibliography{Reference}
\end{document}